\documentclass[pra,twocolumn,aps,amssymb,amsmath,superscriptaddress,showpacs]{revtex4-1}
\usepackage[latin1]{inputenc}
\usepackage{mathrsfs}
\usepackage{dsfont}
\usepackage{amsmath}
\usepackage{amssymb}
\usepackage{amsthm}
\usepackage{amsfonts}
\usepackage{amstext}
\usepackage{amsopn}
\usepackage{amsxtra}
\usepackage{mathrsfs}
\usepackage{dsfont}
\usepackage{graphicx}
\newtheorem{theorem}{Theorem}
\newtheorem{lemma}{Lemma}
\newcommand{\dps}{\displaystyle}
\newcommand{\ii}{\infty}
\newcommand\R{{\ensuremath {\mathbb R} }}

\newcommand\1{{\ensuremath {\mathds 1} }}

\newcommand\nn{\nonumber}

\renewcommand\phi{\varphi}

\newcommand{\cE}{\mathcal{E}}

\newcommand{\cL}{\mathcal{L}}

\newcommand{\bx}{\mathbf{x}}
\newcommand{\by}{\mathbf{y}}
\newcommand{\bz}{\mathbf{z}}
\newcommand{\bu}{\mathbf{u}}
\newcommand{\br}{\mathbf{r}}

\newcommand{\bk}{\mathbf{k}}

\renewcommand{\epsilon}{\varepsilon}

\renewcommand{\geq}{\geqslant}
\renewcommand{\leq}{\leqslant}

\renewcommand{\tilde}{\widetilde}

\newcommand{\dx}{{\rm d}\bx}

\newcommand{\dy}{{\rm d}\by}
\newcommand{\dz}{{\rm d}\bz}
\newcommand{\du}{{\rm d}\bu}
\newcommand{\dk}{{\rm d}\bk}
\newcommand{\dt}{{\rm d}t}
\newcommand{\ds}{{\rm d}s}

\date{\today}

\begin{document}

\title[Lieb-Oxford inequality with gradient correction]{Improved Lieb-Oxford exchange-correlation inequality with gradient correction}

\author{Mathieu Lewin}
\affiliation{CNRS and CEREMADE (UMR CNRS 7534), University of Paris-Dauphine, Place de Lattre de Tassigny, 75775 Paris Cedex 16, France}

\author{Elliott H. Lieb}
\affiliation{Departments of Physics and Mathematics, Princeton University, Jadwin Hall, Washington Road, Princeton, NJ 08544, USA}

\begin{abstract}
We prove a Lieb-Oxford-type inequality on the indirect part of the Coulomb energy of a general many-particle quantum state, with a lower constant than the original statement but involving an additional gradient correction. The result is similar to a recent inequality of Benguria, Bley and Loss~\cite{BenBleLos-12}, except that the correction term is purely local, which is more usual in density functional theory. In an appendix, we discuss the connection between the indirect energy and the classical Jellium energy for constant densities. We show that they differ by an explicit shift due to the long range of the Coulomb potential. 
\end{abstract}

\maketitle

\section{Introduction and main result}

One of the central problems in density functional theory is the estimation of the indirect part $E_{\rm Ind}$ of the Coulomb energy in the ground state. Absent a precise calculation, the next best thing is to have a rigorous bound for $E_{\rm Ind}$ in terms of the density $\rho(\bx)$. Ideally, this estimate should be local, that is, it must be given by an integral of some function of $\rho(\bx)$ and its derivatives, and, possibly, the kinetic energy density~\cite{BacDel-14}.

The indirect part of the Coulomb energy of an $N$-particle quantum state with two-particle spatial density matrix $\Gamma^{(2)}(\bx,\by;\bx',\by')$ (after summation over spins) is defined by
\begin{equation}
E_{\rm Ind}=\int_{\R^3}\int_{\R^3}\frac{\Gamma^{(2)}(\bx,\by;\bx,\by)}{|\bx-\by|}\,\dx\,\dy-D(\rho,\rho)
\end{equation}
where $\dx={\rm d}x\,{\rm d}y\,{\rm d}z$,
\begin{equation}
\rho(\bx)=\frac{2}{N-1}\int_{\R^3}\Gamma^{(2)}(\bx,\by;\bx,\by)\,\dy
\label{eq:def_rho}
\end{equation}
is the corresponding density, with $\int_{\R^3}\rho(\bx)\,\dx=N$,  and
\begin{equation}
D(\rho,\rho)=\frac12 \int_{\R^3}\int_{\R^3}\frac{\rho(\bx)\rho(\by)}{|\bx-\by|}\,\dx\,\dy
\label{eq:def_Coulomb}
\end{equation}
is its (direct) Coulomb energy. The indirect energy $E_{\rm Ind}$ includes all ``exchange'' and ``correlation'' effects~\footnote{In Kohn-Sham theory, the \emph{exchange-correlation energy} $E_{\rm xc}$ by definition also contains a kinetic part. Here $E_{\rm Ind}$ only contains the electrostatic part and it is sometimes denoted as $E_{\rm xc}(\lambda=1)$ in the literature~\cite{LanPer-77}.}. It is considered here for a general $N$-particle state, not necessarily the ground state.
A well known bound on $E_{\rm Ind}$ is the Lieb-Oxford inequality~\cite{LieOxf-80}
\begin{equation}
E_{\rm Ind}\geq -1.68\int_{\R^3}\rho(\bx)^{4/3}\,\dx
\label{eq:usual_LO}
\end{equation}
which was based on an earlier inequality in~\cite{Lieb-79} with the constant 8.52 instead of 1.68, and which was later improved by Chan and Handy to 1.64 in~\cite{ChaHan-99}.

In the original derivation of~\eqref{eq:usual_LO} in~\cite{Lieb-79,LieOxf-80} the bound naturally breaks up into two pieces. The first is
$$-\frac35 \left(\frac{9\pi}{2}\right)^{1/3}\int_{\R^3}\rho(\bx)^{4/3}\,\dx$$
where $(3/5) ({9\pi}/{2})^{1/3}\simeq 1.4508$, and it arises in a straightforward manner, as we will recall below. The second is non local and less transparent, and it requires technical manipulations to estimate it by the simpler local term $\int_{\R^3}\rho(\bx)^{4/3}\,\dx$~\cite{Lieb-79,LieOxf-80,LieSei-09}. A considerable step forward was recently taken by Benguria, Bley and Loss in~\cite{BenBleLos-12} who realized that the second term vanishes if $\rho$ is constant, and that it could be bounded in terms of the quantity $(\sqrt{\rho},|\nabla|\sqrt{\rho})$, defined as
\begin{align}
(\sqrt{\rho},|\nabla|\sqrt{\rho})&=\int_{\R^3}|\bk|\,|\widehat{\sqrt\rho}(\bk)|^2\,\dk\nn\\
&=\frac1{2\pi^2}\int_{\R^3}\int_{\R^3}\frac{\big|\sqrt{\rho(\bx)}-\sqrt{\rho(\by)}\big|^2}{|\bx-\by|^4}\,\dx\,\dy\label{eq:nonlocal}
\end{align}
(here, $\widehat{\sqrt{\rho}}(\bk)$ is the Fourier transform of $\sqrt{\rho(\bx)}$). This quantity has the desired property of being small for slowly varying densities but it is very non local and not easy to compute from the knowledge of $\rho(\bx)$. 
In a numerical code based on Gaussian-type atomic orbitals, computing this term would typically require the use of the Fast Fourier Transform on a very fine grid to compensate for the large variations of $\sqrt{\rho(\bx)}$ in the neighborhood of the atoms, and that could dramatically increase the computational cost.

In this paper we give two alternatives to the bound of Benguria, Bley and Loss, involving only integrals of derivatives of $\rho$. For simplicity, we state the two bounds in one line.\footnote{Here and elsewhere in the paper $\nabla\rho^{\frac13}=\nabla(\rho^{\frac13})$ and 
$|\nabla f(\bx)|=\sqrt{\left|\tfrac{\partial f}{\partial x}(\bx)\right|^2+\left|\tfrac{\partial f}{\partial y}(\bx)\right|^2+\left|\tfrac{\partial f}{\partial z}(\bx)\right|^2}$.}

\begin{theorem}[Lieb-Oxford inequality with gradient correction]
We have the following two bounds
\begin{multline}
E_{\rm Ind} \geq -\left(\frac35 \left(\frac{9\pi}{2}\right)^{\tfrac13}+\alpha\right)\int_{\R^3}\rho(\bx)^{\tfrac43}\,\dx\\
-\begin{cases}
\dps \frac{0.001206}{\alpha^3}\int_{\R^3}|\nabla\rho(\bx)|\,\dx\\
\dps\frac{0.2097}{\alpha^2}\int_{\R^3}|\nabla\rho^{\frac13}(\bx)|^2\,\dx
\end{cases}
\label{eq:LO_gradient}
\end{multline}
for all $\alpha>0$.
After optimizing over $\alpha$ in~\eqref{eq:LO_gradient}, an equivalent formulation is
\begin{multline}
E_{\rm Ind} \geq -\frac35 \left(\frac{9\pi}{2}\right)^{\tfrac13}\int_{\R^3}\rho(\bx)^{\tfrac43}\,\dx\\
-\begin{cases}
\dps0.3270\left(\int_{\R^3}|\nabla\rho(\bx)|\,\dx\right)^{\tfrac14}\left(\int_{\R^3} \rho(\bx)^{\tfrac43}\,\dx\right)^{\tfrac34}\\
\dps1.1227\left(\int_{\R^3}|\nabla\rho^{\frac13}(\bx)|^2\,\dx\right)^{\tfrac13}\left(\int_{\R^3} \rho(\bx)^{\tfrac43}\,\dx\right)^{\tfrac23}\!\!.
\end{cases}
\label{eq:LO_gradient2}
\end{multline}
\end{theorem}

The inequalities are very similar to that of~\cite{BenBleLos-12}, except that they only involve local quantities. If the gradient term is very small, then only the first term matters in~\eqref{eq:LO_gradient2} and the constant is approximately $1.4508$. 
For classical Jellium in infinite volume, the lower bound $-1.4508\,\rho^{4/3}$ on the energy per unit volume has already been proved in~\cite{LieNar-75}. The classical Jellium energy is at most $-1.4442\,\rho^{4/3}$, which is computed by placing the particles on the BCC lattice, as suggested first by Wigner~\cite{Wigner-34}, and is remarkably near $-1.4508\,\rho^{4/3}$.
It is sometimes suggested in the literature that the same holds for the indirect energy~\cite{Perdew-91,PerWan-92,LevPer-93,OdaCap-07,RasPitCapPro-09,GiuVig-05}. However, the link between $E_{\rm Ind}$ and Jellium appears to be more involved than expected and we discuss this more fully in Appendix~\ref{app:Wigner}. The lowest value that we can rigorously achieve for the indirect energy at constant density is $-0.9507\rho^{4/3}$.

Inserting $|\nabla \rho(\bx)|=3\rho(\bx)^{2/3}|\nabla\rho^{1/3}(\bx)|$
in the first bound in~\eqref{eq:LO_gradient2} and using Schwarz's inequality, one immediately gets 
 \begin{multline}
 E_{\rm Ind} \geq -\frac35 \left(\frac{9\pi}{2}\right)^{\tfrac13}\int_{\R^3}\rho(\bx)^{\tfrac43}\,\dx\\
-0.4304\left(\int_{\R^3}|\nabla\rho^{\frac13}(\bx)|^2\,\dx\right)^{\tfrac18}\left(\int_{\R^3} \rho(\bx)^{\tfrac43}\,\dx\right)^{\tfrac78}.
\label{eq:LO_gradient3}
\end{multline}
The constant is smaller than in~\eqref{eq:LO_gradient2}, but the gradient
term has the power $1/8$ instead of $1/3$. 
Taking (\ref{eq:LO_gradient2}) line 2 to the power 3/5 times the last term of 
(\ref{eq:LO_gradient3}) to the power 2/5 gives the bound
\begin{multline}
 E_{\rm Ind} \geq -\frac35
\left(\frac{9\pi}{2}\right)^{\tfrac13}\int_{\R^3}\rho(\bx)^{\tfrac43}\,\dx\\
-0.7651\left(\int_{\R^3}|\nabla\rho^{\frac13}(\bx)|^2\,\dx\right)^{\tfrac14}
\left(\int_{\R^3} \rho(\bx)^{\tfrac43}\,\dx\right)^{\tfrac34}.
\label{eq:LO_gradient4}
\end{multline}
which is comparable in form to (\ref{eq:LO_gradient2}) line 1. 

Our proof for the second bound in~\eqref{eq:LO_gradient2} is more involved
and relies on an inequality for a Hardy-Littlewood-type maximal function
(see~\eqref{eq:HL_ineq_chi} in Lemma~\ref{lem:HL_chi} below). We can expect
to improve our constants when improvements to~\eqref{eq:HL_ineq_chi} are
found. We emphasize that our first bound in~\eqref{eq:LO_gradient2} does not involve anything as complicated as the maximal function.

We remark that the bounds in~\eqref{eq:LO_gradient} and~\eqref{eq:LO_gradient2} apply equally for bosons and fermions, as is already the case for~\eqref{eq:usual_LO} and all the other existing estimates of this kind. It is a challenging open problem to derive a better estimate taking the fermionic statistics into account~\cite{BacDel-14}.

The local term $\int_{\R^3}|\nabla\rho(\bx)|\,\dx$ has the same scaling behavior as the non local term $(\sqrt\rho,|\nabla|\sqrt\rho)$ used by Benguria, Bley and Loss in~\cite{BenBleLos-12}, and it seems a natural alternative. We remark that, by the theorem of M. and T. Hoffmann-Ostenhof in~\cite{Hof-77}, the local gradient
$|\nabla\rho(\bx)|$ can be controlled by the local kinetic energy density of the $N$ particles defined by~\cite{BacDel-14} 
\begin{equation}
\tau(\bx)=\frac{2}{N-1}\int_{\R^3}\big[\nabla_\bx\cdot\nabla_{\bx'}\Gamma^{(2)}\big](\bx,\by;\bx,\by)\,\dy,
\label{eq:def_t}
\end{equation}
as follows:
$$|\nabla\rho(\bx)|=2\sqrt{\rho(\bx)}|\nabla\sqrt{\rho}(\bx)|\leq 2\sqrt{\rho(\bx)}\;\sqrt{\tau(\bx)}.$$
In particular, by Schwarz's inequality, $\int_{\R^3}|\nabla\rho(\bx)|\,\dx$ is always finite if the total quantum kinetic energy of the $N$ particles is finite.

On the other hand, the gradient term $|\nabla\rho^{1/3}(\bx)|^2$ is a well-known second-order correction to the correlation energy in the regime of a high and slowly varying density~\cite{MaBru-68}. This term is also often used in density functional theory, especially since the appearance of the highly cited paper~\cite{PerBurErn-96} of Perdew, Burke and Ernzerhof, who proposed the use of the (PBE) exchange correlation functional\footnote{The functional is written here for a system of unpolarized electrons, for simplicity.}
\begin{multline*}
-\frac{3^{\frac43}}{4\pi^{\frac13}}\int_{\R^3}\rho(\bx)^{\frac43}\left(1+\frac{\mu|\nabla\rho^{\frac13}(\bx)|^2}{\rho(\bx)^{\frac43}+(\mu/\nu)|\nabla\rho^{\frac13}(\bx)|^2}\right)\,\dx\\
+\int_{\R^3}\rho(\bx)\,\varepsilon_{\rm c}\big(\rho(\bx)\big)\,\dx\\
+\gamma\int_{\R^3}\dx\, \rho(\bx)\log\Bigg(1+\frac{\beta}{\gamma}|\nabla\rho^{\frac13}(\bx)|^2\times\\
\times\frac{\rho(\bx)+A\big(\rho(\bx)\big)|\nabla\rho^{\frac13}(\bx)|^2}{\rho(\bx)^2+A\big(\rho(\bx)\big)\rho(\bx)|\nabla\rho^{\frac13}(\bx)|^2+A\big(\rho(\bx)\big)^2|\nabla\rho^{\frac13}(\bx)|^4}\Bigg).
\end{multline*}
Here $\varepsilon_{\rm c}\big(\rho(\bx)\big)$ is the (unknown~\cite{PerWan-92}) correlation energy of the interacting Coulomb gas, 
$$A\big(\rho(\bx)\big)=\frac{\beta}{\gamma}\left(e^{-{\varepsilon_{\rm c}(\rho(\bx))}/\gamma}-1\right)^{-1}$$
and $\mu,\nu,\gamma,\beta$ are constants which are uniquely determined from a list of reasonable physical requirements. In particular, the value of $\nu$ was chosen in relation with the original Lieb-Oxford inequality~\eqref{eq:usual_LO}. Our bound could be used to improve the existing density functionals.

For large atoms where the quantum mechanical density $\rho(\bx)$ is known to behave semi-classically~\cite{Lieb-81b}, it was noticed in~\cite[Appendix]{BenBleLos-12} that the non-local correction term~\eqref{eq:nonlocal} is much smaller than the first term, and the situation is the same in our case. Indeed, if we plug in the ansatz $\rho(\bx)=Z^2\, \tilde\rho(Z^{1/3}\bx)$ with a fixed density $\tilde\rho$, as is the case 
in Thomas-Fermi theory, then we find
\begin{align*}
\int_{\R^3}\rho(\bx)^{4/3}\,\dx&=Z^{5/3}\int_{\R^3}\tilde\rho(\bx)^{4/3}\,\dx,\\
(\sqrt{\rho},|\nabla|\sqrt{\rho})&=Z^{4/3}(\sqrt{\tilde\rho},|\nabla|\sqrt{\tilde\rho}),\\
\int_{\R^3}|\nabla\rho(\bx)|\,\dx&=Z^{4/3}\int_{\R^3}|\nabla\tilde\rho(\bx)|\,\dx,\\
\int_{\R^3}|\nabla\rho^{1/3}(\bx)|^2\,\dx&=Z\int_{\R^3}|\nabla\tilde\rho^{1/3}(\bx)|^2\,\dx.
\end{align*}
In the regime $Z\gg1$, the term involving $|\nabla\rho^{1/3}(\bx)|^2$ is one order lower than $|\nabla\rho(\bx)|$. However, we remark that, in contrast to $|\nabla\tilde\rho(\bx)|$ which is always controlled by the kinetic energy density, for a general quantum state $|\nabla\tilde\rho^{1/3}(\bx)|^2$ can be very large. If we use the Thomas-Fermi ground state density $\tilde\rho(\bx)=\rho_{\rm TF}(\bx)$, the integral $\int|\nabla\tilde\rho^{1/3}(\bx)|^2\,\dx$ actually diverges at $\bx=0$~\cite{MaBru-68}. This is not devastating since the Thomas-Fermi density must be modified at small distances by the Scott correction. Nevertheless, this suggests that $|\nabla\rho(\bx)|$ might be a less sensitive alternative in density functional theory.

The rest of the paper is devoted to the proof of our main result.

\section{Proof of the two inequalities}

Unless specified, all the integrals are over the Euclidean space $\R^3$.
Our starting point is the following inequality, taken from the Lieb-Oxford paper~\cite{LieOxf-80},
\begin{multline}
E_{\rm Ind}\geq -\frac{3}{5}\left(\frac{9\pi}{2}\right)^{1/3}\int_{\R^3}\rho(\bx)^{4/3}\,\dx\\+2\int_{\R^3}\rho(\bx)D\big(\rho-\rho(\bx),\mu_\bx-\delta_\bx\big)\,\dx,
\label{eq:starting}
\end{multline}
where 
\begin{align}
\mu_\bx(\by)&:=\frac{3}{4\pi R(\bx)^{3}}\,\1\big(|\by-\bx|\leq R(\bx)\big)\nn\\
&=\rho(\bx)\, \1\left(|\by-\bx|^3\leq\frac{3}{4\pi\rho(\bx)}\right)
\label{eq:def_mu}
\end{align}
is the normalized uniform measure of the ball centered at $\bx$ with radius
\begin{equation}
R(\bx):=\left(\frac{3}{4\pi\rho(\bx)}\right)^{1/3}.
\end{equation}
In Eq.~\eqref{eq:def_mu}, the notation $\1\big(|\by-\bx|\leq R(\bx)\big)$ means 
$$\1\big(|\by-\bx|\leq R(\bx)\big)=\begin{cases}
1 &\text{ if $|\by-\bx|\leq R(\bx)$,}\\
0 &\text{ otherwise,}
\end{cases}
$$
and a similar convention is used for all the other expressions involving $\1$. The derivation of~\eqref{eq:starting} is outlined in Appendix~\ref{app:derivation} for the convenience of the reader. We denote the correction term by
$$\text{Corr}:=2\int_{\R^3}\rho(\bx)D\big(\rho-\rho(\bx),\mu_\bx-\delta_\bx\big)\,\dx$$
and our goal is to find a lower bound to it using the gradient of $\rho(\bx)$.

The Coulomb potential of a point charge screened by a uniform distribution in a ball is $(r^2/2+1/r-3/2)\1(0\leq r\leq1):=-\Phi(r)$. Introducing the function
$$0\leq \Psi(r)=-r^4\Phi(r)=r^4\left(\frac{r^2}2+\frac{1}{r}-\frac{3}{2}\right)\1(0\leq r\leq1),$$
we can then write the correction term as
\begin{align}
\text{Corr}&=\frac{3}{4\pi}\iint\frac{\rho(\bx)-\rho(\by)}{|\bx-\by|^4}\Psi\left(\frac{|\bx-\by|}{R(\bx)}\right)\,\dx\,\dy\nn\\
&=\frac{3}{8\pi}\iint\frac{\rho(\bx)-\rho(\by)}{|\bx-\by|^4}\times\nn\\
&\quad\times\left(\Psi\left(\frac{|\bx-\by|}{R(\bx)}\right)-\Psi\left(\frac{|\bx-\by|}{R(\by)}\right)\right)\dx\,\dy.
\label{eq:error_term}
\end{align}
In the second line we have simply exchanged the role of $\bx$ and $\by$ to get a more symmetric expression. 

As remarked in~\cite{LieOxf-80}, the function $\Psi$ is first increasing on $[0,r_*]$ and then decreasing on $[r_*,1]$ with $r_*=(\sqrt{5}-1)/2$. Hence we can write $\Psi$ as a difference of two non-negative increasing functions as follows
$$\Psi=\int_0^r(\Psi')_+-\int_0^r(\Psi')_-:=\Psi_1-\Psi_2,$$
where $(\Psi')_+:=\max(\Psi',0)$ and $(\Psi')_-:=\max(-\Psi',0)$ are, respectively, the positive and negative parts of $\Psi'$. Note that $\Psi_1$ and $\Psi_2$ are constant for $r\geq1$, equal to $\Psi(r_*)\simeq0.04509$. 
From the monotonicity of $\Psi_1$, we deduce that
\begin{multline}
\text{Corr}\geq -\frac{3}{8\pi}\iint\frac{\rho(\bx)-\rho(\by)}{|\bx-\by|^4}\times\\
\times\left(\Psi_2\left(\frac{|\bx-\by|}{R(\bx)}\right)-\Psi_2\left(\frac{|\bx-\by|}{R(\by)}\right)\right)\,\dx\,\dy.
\label{eq:to_be_studied}
\end{multline}
Since $\Psi_2(r)=0$ for $r\leq r_*$, we may always add the constraint that
\begin{multline*}
|\bx-\by|\max\left(R(\bx)^{-1},R(\by)^{-1}\right)\\=\frac{|\bx-\by|}{c}\max\left(\rho(\bx),\rho(\by)\right)^{1/3}\geq r_*, 
\end{multline*}
where $c=({3}/{4\pi})^{1/3}$. We now introduce a parameter $\theta>r_*$ and split the right side of~\eqref{eq:to_be_studied} as a sum of two terms:
\begin{multline}
\text{Corr}_1=-\frac{3}{8\pi}\iint_{\max(\rho(\bx),\rho(\by))^{1/3}> \frac{\theta c}{|\bx-\by|}}\frac{\rho(\bx)-\rho(\by)}{|\bx-\by|^4}\times\\
\times\left(\Psi_2\left(\frac{|\bx-\by|}{R(\bx)}\right)-\Psi_2\left(\frac{|\bx-\by|}{R(\by)}\right)\right)\,\dx\,\dy
\label{eq:Corr1}
\end{multline}
that will be bounded in terms of $\theta^{-1}\int\rho(\bx)^{4/3}\,\dx$, and
\begin{multline}
\text{Corr}_2=\\-\frac{3}{8\pi}\iint_{\frac{r_*c}{|\bx-\by|}\leq \max(\rho(\bx),\rho(\by))^{1/3}\leq \frac{\theta c}{|\bx-\by|}}\frac{\rho(\bx)-\rho(\by)}{|\bx-\by|^4}\times\\ \times\left(\Psi_2\left(\frac{|\bx-\by|}{R(\bx)}\right)-\Psi_2\left(\frac{|\bx-\by|}{R(\by)}\right)\right)\,\dx\,\dy.
\label{eq:Corr2}
\end{multline}
that will be bounded in two different ways, using $\theta^{3}\int|\nabla\rho(\bx)|\,\dx$ and $\theta^2\int|\nabla\rho^{1/3}(\bx)|^2\,\dx$. The parameter $\theta$ will be chosen proportional to $1/\alpha$ which appears in~\eqref{eq:LO_gradient}.

\subsection{Bound on $\text{Corr}_1$ in terms of $\rho(\bx)^{4/3}$}
In order to derive a lower bound on $\text{Corr}_1$, we simply use the fact that $\Psi_2\leq \Psi(r_*)$ and obtain
\begin{align*}
\text{Corr}_1&\geq -\frac{3}{4\pi}\Psi(r_*)\iint_{\substack{\rho(\bx)^{1/3}> \frac{\theta c}{|\bx-\by|}\\ \rho(\by)\leq\rho(\bx)}}\frac{\rho(\bx)}{|\bx-\by|^4} \,\dx\,\dy\\
&\geq -3\Psi(r_*)\int\rho(\bx)\left(\int_{\frac{\theta c}{\rho(\bx)^{1/3}}}^\ii\frac{{\rm d}r}{r^2}\right)\,\dx\\
&= -\frac{3\Psi(r_*)}{\theta c}\int\rho(\bx)^{4/3}\,\dx.
\end{align*}
By defining 
$\alpha={3\Psi(r_*)}/{\theta c}\simeq {0.2180}/{\theta},$
this gives the first part of the correction term in~\eqref{eq:LO_gradient}. Note that we have assumed that $\theta>r_*$, which yields the condition $\alpha\leq 3\Psi(r_*)/(cr_*)\simeq0.3528$. For $\alpha>0.3528$, our bound~\eqref{eq:LO_gradient} is already worse than the Lieb-Oxford bound~\eqref{eq:usual_LO} and there is nothing to prove.

\subsection{Bound on $\text{Corr}_2$ in terms of $|\nabla\rho(\bx)|$}
In the formula of ${\rm Corr}_2$, we insert the fundamental theorem of calculus in the form
\begin{multline*}
\Psi_2\left(\frac{|\bx-\by|\rho(\bx)^{\frac13}}{c}\right)-\Psi_2\left(\frac{|\bx-\by|\rho(\by)^{\frac13}}{c}\right)\\
=\frac{|\bx-\by|}{c}\int_0^1(\bx-\by)\cdot \nabla\rho^{\frac13}(\by+t(\bx-\by))\times\\
\times\Psi_2'\left(\frac{|\bx-\by|\rho(\by+t(\bx-\by))^{\frac13}}{c}\right)\,\dt
\end{multline*}
and obtain
\begin{align*}
\text{Corr}_2
&\geq -\frac{c^2}{2}\int_0^1\dt\iint_{\frac{r_*c}{|\bx-\by|}\leq \max(\rho(\bx),\rho(\by))^{1/3}\leq \frac{\theta c}{|\bx-\by|}}\\
&\qquad\frac{\max(\rho(\bx),\rho(\by))}{|\bx-\by|^2}|\nabla\rho^{\frac13}(\by+t(\bx-\by))|\times\\
&\qquad\times
\Psi_2'\left(\frac{|\bx-\by|\rho(\by+t(\bx-\by))^{\frac13}}{c}\right)\,\dx\,\dy.
\end{align*}
Since $\Psi_2'$ is supported in $[r_*,1]$, we have
$$\frac{cr_*}{|\bx-\by|}\leq \rho(\by+t(\bx-\by))^{\frac13}\leq \frac{c}{|\bx-\by|},\qquad\forall\; 0\leq t\leq1.$$
From the constraint in the integral, we therefore get the relation
\begin{equation}
\max(\rho(\bx),\rho(\by))^{1/3}\leq \frac{\theta c}{|\bx-\by|}\leq \frac{\theta}{r_*}\rho(\by+t(\bx-\by))^{\frac13}
\label{eq:relation_rho}
\end{equation}
and hence
\begin{multline*}
\text{Corr}_2\geq -\frac{c^2}{2(r_*)^3}\theta^3\int_0^1\dt\iint\frac{\rho(\by+t(\bx-\by))}{|\bx-\by|^2}\times\\
\times|\nabla\rho^{\frac13}(\by+t(\bx-\by))|\times\\ \times\Psi_2'\left(\frac{|\bx-\by|\rho(\by+t(\bx-\by))^{\frac13}}{c}\right)\,\dx\,\dy.
\end{multline*}
Introducing $\bz=\by+t(\bx-\by)$, this is the same as
\begin{multline*}
\text{Corr}_2\geq-\frac{c^2}{2(r_*)^3}\theta^3\iint\frac{\rho(\bz)}{|\bz-\by|^2}\times\\
\times|\nabla\rho^{\frac13}(\bz)|\left(\int_0^1\Psi_2'\left(\frac{|\bz-\by|\rho(\bz)^{\frac13}}{ct}\right)\frac{\dt}{t}\right)\,\dz\,\dy.
\end{multline*}
To deal with the term in the parenthesis, we note that
$$\int_0^1\Psi_2'\left(\frac{A}{t}\right)\frac{\dt}{t}=\int_A^{\ii}\Psi_2'\left(t\right)\frac{\dt}{t}\leq \1(A\leq 1)\int_{r_*}^{1}\Psi'_2(t)\frac{\dt}{t}$$
for every $A\geq0$, and obtain
\begin{align*}
\text{Corr}_2&\geq -\frac{c^2}{2(r_*)^3}\theta^3\left(\int_{r_*}^{1}\frac{\Psi_2'(t)}{t}\dt\right)\times\\
&\qquad\times\iint_{|\bz-\by|\leq c\rho(\bz)^{-1/3}}\frac{\rho(\bz)}{|\bz-\by|^2}
|\nabla\rho^{\frac13}(\bz)|\,\dz\,\dy\\
&= -\frac{2\pi c^3}{(r_*)^3}\theta^3\left(\int_{r_*}^{1}\frac{\Psi_2'(t)}{t}\dt\right)\int\rho(\bz)^{2/3}
|\nabla\rho^{\frac13}(\bz)|\,\dz\\
&= -\frac{2\pi c^3}{3(r_*)^3}\theta^3\left(\int_{r_*}^{1}\frac{\Psi_2'(t)}{t}\dt\right)\int|\nabla\rho(\bz)|\,\dz.
\end{align*}
Inserting the definition of $\alpha$, we find
\begin{equation}
\text{Corr}_2\geq -\frac{18\pi \Psi(r_*)^3}{(r_*)^3\alpha^3}\left(\int_{r_*}^{1}\frac{\Psi_2'(t)}{t}\dt\right)\int|\nabla\rho(\bz)|\,\dz
\end{equation}
which concludes the derivation of the first inequality in~\eqref{eq:LO_gradient} since
$$\frac{18\pi \Psi(r_*)^3}{(r_*)^3}\left(\int_{r_*}^{1}\frac{\Psi_2'(t)}{t}\dt\right)\simeq0.001206.$$

\subsection{Bound on $\text{Corr}_2$ in terms of $|\nabla\rho^{1/3}(\bx)|^2$}
Our second bound on ${\rm Corr}_2$ is more involved and it relies on the theory of Hardy-Littlewood maximal functions~\cite{Grafakos-book}. 
First we use
$$|\rho(\bx)-\rho(\by)|\leq 3\max(\rho(\bx)^{\frac23},\rho(\by)^{\frac23})\big|\rho(\bx)^{\frac13}-\rho(\by)^{\frac13}\big|$$
and then appeal twice to the fundamental theorem of calculus. We find
\begin{multline*}
\text{Corr}_2\geq\\
-\frac{3c^2}{2}\int_0^1\dt\int_0^1\ds\iint_{\frac{r_*c}{|\bx-\by|}\leq \max(\rho(\bx),\rho(\by))^{1/3}\leq \frac{\theta c}{|\bx-\by|}}\\ \frac{\max(\rho(\bx),\rho(\by))^{2/3}}{|\bx-\by|}|\nabla\rho^{\frac13}(\by+s(\bx-\by))|\times\\
\times|\nabla\rho^{\frac13}(\by+t(\bx-\by))|\times\\ \times \Psi_2'\left(\frac{|\bx-\by|\rho(\by+t(\bx-\by))^{\frac13}}{c}\right)\,\dx\,\dy.
\end{multline*}
As before, we may use the constraint~\eqref{eq:relation_rho} and obtain
\begin{multline*}
\text{Corr}_2\geq\\ -\frac{3\theta^2c^2}{2(r_*)^2}\int_0^1\dt\int_0^1\ds\iint\frac{\rho(\by+t(\bx-\by))^{2/3}}{|\bx-\by|}\times\\ \times|\nabla\rho^{\frac13}(\by+s(\bx-\by))|\;|\nabla\rho^{\frac13}(\by+t(\bx-\by))|\times\\ \times \Psi_2'\left(\frac{|\bx-\by|\rho(\by+t(\bx-\by))^{\frac13}}{c}\right)\,\dx\,\dy
\end{multline*}
or equivalently, after a change of variables,
\begin{multline}
\text{Corr}_2\geq-\frac{3\theta^2c^4}{2(r_*)^2}\int_0^1\dt\int_0^1\,\ds\int\dz\; |\nabla\rho^{\frac13}(\bz)|\times\\
\times\left(\frac{\rho(\bz)^{2/3}}{c^2|t-s|^2}\int\frac{|\nabla\rho^{\frac13}(\bz+\bu)|}{|\bu|}
\; \Psi_2'\left(\frac{|\bu|\rho(\bz)^{\frac13}}{c|t-s|}\right)\,\du\right).\label{eq:to_be_estimated}
\end{multline}
In order to control the term in the parenthesis, we introduce the positive function $\chi(r)=r^{-1}\Psi_2'(r)=-r^{-1}\Psi'(r)\1(r_*\leq r\leq 1)$ which is increasing on $[r_*,s_{*}]$ and decreasing on $[s_{*},1]$, with $s_*\simeq 0.8376$. For every square-integrable function $f$, we define the following Hardy-Littlewood-type function
\begin{equation}
{\rm M}^\chi_f(\bz):=\sup_{r>0}\left\{r^{-3}\int \chi(|\bu|/r)\,|f(\bz+\bu)|\,\du\right\}.
\label{eq:HL-chi}
\end{equation}
From this definition we have 
\begin{multline*}
\frac{\rho(\bz)^{2/3}}{c^2|t-s|^2}\int\frac{|\nabla\rho^{\frac13}(\bz+\bu)|}{|\bu|}
\; \Psi_2'\left(\frac{|\bu|\rho(\bz)^{\frac13}}{c|t-s|}\right)\,\du\\
\leq {\rm M}^\chi_{|\nabla\rho^{1/3}|}(\bz) 
\end{multline*}
for all $\bz\in\R^3$ and hence, by Schwarz's inequality,
\begin{align*}
\text{Corr}_2&\geq-\frac{3\theta^2c^4}{2(r_*)^2}\int |\nabla\rho^{\frac13}(\bz)|\, {\rm M}^\chi_{|\nabla\rho^{1/3}|}(\bz)\,\dz\\
&\geq-\frac{3\theta^2c^4}{2(r_*)^2}\left(\int |\nabla\rho^{\frac13}(\bz)|^2\,\dz\right)^{\frac12}\times\\
&\qquad\times\left(\int {\rm M}^\chi_{|\nabla\rho^{1/3}|}(\bz)^2\,\dz\right)^{\frac12}.
\end{align*}

\begin{lemma}[Hardy-Littlewood-type inequality]\label{lem:HL_chi}
For every square-integrable function $f$ on $\R^3$, we have
\begin{equation}
\left(\int {\rm M}^\chi_{f}(\bz)^2\,\dz\right)^{1/2}\leq 7.5831\left(\int |f(\bz)|^2\,\dz\right)^{1/2}.
\label{eq:HL_ineq_chi}
\end{equation}
\end{lemma}

The optimal constant in~\eqref{eq:HL_ineq_chi} is not known and the value $7.5831$ is based on a numerical calculation explained below. In the proof of Lemma~\ref{lem:HL_chi} below we also provide a simple explicit bound which yields the slightly worse constant $8.2163$.

Using~\eqref{eq:HL_ineq_chi} with $f=|\nabla\rho^{1/3}|$ immediately gives our final bound
\begin{align*}
\text{Corr}_2&\geq-\frac{3\theta^2c^4}{2(r_*)^2}\cdot7.5831 \int |\nabla\rho^{\frac13}(\bz)|^2\,\dz\\
&=-\underbrace{\frac{27\Psi(r_*)^2c^2}{2(r_*)^2\alpha^2}\cdot7.5831}_{\simeq \frac{0.2097}{\alpha^2}} \int |\nabla\rho^{\frac13}(\bz)|^2\,\dz.
\end{align*}
This concludes our study of the correction term ${\rm Corr}_2$ and the derivation of the second inequality in~\eqref{eq:LO_gradient}. The rest of this section will be devoted to the 

\begin{proof}[Proof of Lemma~\ref{lem:HL_chi}]
For the convenience of the reader, we explain first an easier calculation yielding the constant $8.2163$. We write, for every $r>0$,
\begin{align*}
&r^{-3}\int \chi(|\bu|/r)|f(\bz+\bu)|\,\du\\
&=-\int_{r_*}^1 s^3\chi'(s)\left((rs)^{-3}\!\int|f(\bz+\bu)|\,\1(|\bu|\leq rs)\,\du\right)\!{\rm d}s\\
&\leq\int_{s_*}^{1} s^3|\chi'(s)|\left((rs)^{-3}\int|f(\bz+\bu)|\,\1(|\bu|\leq rs)\,\du\right)\!{\rm d}s.
\end{align*}
We recall that $\chi'\geq0$ on $[r_*,s_*]$ and $\chi'\leq0$ on $[s_*,1]$, with $s_*\simeq0.8376$.
The expression in the parenthesis can be controlled in terms of the usual Hardy-Littlewood maximal function~\cite{Grafakos-book} defined as
$${\rm M}_{f}(\bz):=\max_{A>0}\left\{\frac{3}{4\pi A^3}\int|f(\bz+\bu)|\,\1(|\bu|\leq A)\,\du\right\},$$
leading to the bound
$${\rm M}^\chi_f(\bz)\leq \frac{4\pi}{3}\left(\int_{s_*}^{1} s^3|\chi'(s)|\,{\rm d}s\right){\rm M}_{f}(\bz).$$
The Hardy-Littlewood theorem states that
\begin{equation}
\int{\rm M}_{f}(\bz)^2\,\dz\leq 4\times3^{3}\int|f(\bz)|^2\,\dz
\label{eq:Littlewood}
\end{equation}
(see, e.g.,~\cite[Thm. 2.1.6]{Grafakos-book}), and we obtain
\begin{multline*}
\left(\int {\rm M}^\chi_{f}(\bz)^2\,\dz\right)^{1/2}\\ \leq \underbrace{8\pi\sqrt{3}\left(\int_{s_*}^{1} s^3|\chi'(s)|\,{\rm d}s\right)}_{\simeq 8.2163} \left(\int |f(\bz)|^2\,\dz\right)^{1/2}. 
\end{multline*}

In order to get the better constant $7.5831$, we use another maximal function based on the time-averaged heat kernel
\begin{align*}
\tilde{\rm M}_{f}(\bz)&:=\sup_{T>0} \left\{\frac1T\int_0^T\left(\int \frac{e^{-\frac{|\bu|^2}{4t}}}{(4\pi t)^{3/2}}\,|f(\bz+\bu)|\,\du\right)\,{\rm d}t\right\}\\ &=\sup_{T>0} \left\{\int\left(\int_{\frac{|\bu|}{2\sqrt{T}}}^\ii \frac{e^{-s^2}}{2\pi^{3/2}|\bu|T}{\rm d}s\right)|f(\bz+\bu)|\,\du\right\}.
\end{align*}
A consequence of the Hopf-Dunford-Schwartz ergodic theorem~\cite[VIII.7.7]{DunfordSchwartz-1} is that
\begin{equation}
\int\tilde {\rm M}_f(\bz)^2\,\dz\leq 8\int|f(\bz)|^2\,\dz,
\label{eq:HDS}
\end{equation}
for every square-integrable function $f$, which is much better than~\eqref{eq:Littlewood}. For every fixed $T>0$, we introduce
$${\rm K}(T)=2\pi^{3/2}T\sup_{r_*\leq r\leq1}\left\{\frac{r\chi(r)}{\dps\int_{\frac{r}{2\sqrt{T}}}^\ii e^{-s^2}{\rm d}s}\right\}.$$
Then 
\begin{equation}
\chi(r)\leq {\rm K}(T)\int_{\frac{r}{2\sqrt{T}}}^\ii \frac{e^{-s^2}}{2\pi^{3/2}rT}{\rm d}s
\label{eq:estim_chi_heat}
\end{equation}
and, therefore, ${\rm M}^\chi_f\leq {\rm K}(T)\,\tilde{\rm M}_f$. Using~\eqref{eq:HDS}, we find
\begin{multline*}
\left(\int {\rm M}^\chi_{f}(\bz)^2\,\dz\right)^{1/2}\\ \leq 2\sqrt{2}\left(\min_{T>0}{\rm K}(T)\right) \left(\int |f(\bz)|^2\,\dz\right)^{1/2}. 
\end{multline*}
A numerical calculation gives that
$\min_{T>0}{\rm K}(T)\simeq2.68102$
which is attained for $T\simeq0.2762$, leading to the stated bound with the constant $7.5831$.
\end{proof}

\appendix
\section{Derivation of~\eqref{eq:starting}}\label{app:derivation}

We outline here the derivation of Eq.~\eqref{eq:starting}, for the convenience of the reader. We follow the original approach of Lieb and Oxford in~\cite{LieOxf-80} which is also explained in~\cite[Chap. 6]{LieSei-09}.

As before let 
\begin{multline}
\mu_\bx(\by):=\frac{3}{4\pi R(\bx)^{3}}\1\big(|\by-\bx|\leq R(\bx)\big),\\ \text{with}\quad R(\bx):=\frac{c}{\rho(\bx)^{1/3}},
\end{multline}
be the uniform measure of the ball $B_\bx$ centered at $\bx$, of radius $R(\bx)$. At this step the constant $c$ is arbitrary and we will explain later how the choice $c=(3/(4\pi))^{1/3}$ made in the text arises naturally.
An inequality due to Onsager~\cite{Onsager-39} (see~\cite[page 113]{LieSei-09}) allows us to control the many-particle Coulomb energy by a one-particle term as follows:
\begin{align*}
\sum_{1\leq k<\ell\leq N}\frac{1}{|\bx_k-\bx_\ell|}&\geq \sum_{1\leq k\neq \ell\leq N}D(\mu_{\bx_k},\mu_{\bx_\ell})\\
&\geq -D(\rho,\rho)+2\sum_{i=1}^N D(\rho,\mu_{\bx_i})\\
&\qquad\qquad-\sum_{i=1}^N D(\mu_{\bx_i},\mu_{\bx_i}).
\end{align*}
Integrating over the $N$-particle quantum state we obtain
\begin{align*}
E_{\rm Ind}&=\int\int\frac{\Gamma^{(2)}(\bx,\by;\bx,\by)}{|\bx-\by|}\,\dx\,\dy-D(\rho,\rho)\\
&\geq 2\int\rho(\bx)D(\rho,\mu_\bx-\delta_\bx)\,\dx-\int\rho(\bx)D(\mu_\bx,\mu_\bx)\,\dx.
\end{align*}
At this point this is exactly~\cite[Eq. (17)]{LieOxf-80} and~\cite[Eq. (6.4.8)]{LieSei-09}. 

Note that the Coulomb potential generated by the measure $\mu_\bx-\delta_\bx$ vanishes outside of $B_\bx$, and we can therefore subtract the constant $\rho(\bx)$ in the first term. We obtain 
\begin{multline*}
E_{\rm Ind}\geq \int\rho(\bx)^2D(1,\mu_\bx-\delta_\bx)\,\dx\\+2\int\rho(\bx)D\big(\rho-\rho(\bx),\mu_\bx-\delta_\bx\big)\,\dx-\int\rho(\bx)D(\mu_\bx,\mu_\bx)\,\dx.
\end{multline*}
By scaling we have
$$D(\mu_\bx,\mu_\bx)=\frac{D(\mu,\mu)}{R(\bx)}$$
and 
$$D(\1_{B_\bx},\mu_\bx-\delta_\bx)= R(\bx)^2 D(\1_{B},\mu-\delta_0),$$
where $\mu$ is the uniform measure of the unit ball $B$, with volume $|B|=4\pi/3$. This gives us 
\begin{multline*}
E_{\rm Ind}\geq \left(2c^2D(\1_{B},\mu-\delta_0)-\frac{D(\mu,\mu)}{c}\right)\int\rho(\bx)^{4/3}\,\dx\\
+2\int\rho(\bx)D\big(\rho-\rho(\bx),\mu_\bx-\delta_\bx\big)\,\dx.
\end{multline*}
We have 
\begin{align*}
D(\mu,\mu)&=\frac{|B|^{-2}}{2}\int_B\int_B\frac{\dx\,\dy}{|\bx-\by|}\\
&=\frac{|B|^{-2}}{2}\int_B\int_B\frac{\dx\,\dy}{\max(|\bx|,|\by|)}\\
&=|B|^{-2}(4\pi)^2\int_0^1rdr\int_0^r s^2ds=\frac35
\end{align*}
and
\begin{align*}
&2c^2D(\1_{B},\mu-\delta_0)-\frac{D(\mu,\mu)}{c}\\
&\qquad=2c^2D(\1_{B},\mu)-\frac{D(\mu,\mu)}{c}-c^2\int_{B}\frac{\dx}{|\bx|}\\
&\qquad=\left(\frac{8\pi c^2}{3}-\frac1c\right)D(\mu,\mu)-c^22\pi\\
&\qquad=-\frac{2\pi c^2}{5}-\frac{3}{5c}.
\end{align*}
Optimizing over $c$ gives $c=(3/(4\pi))^{1/3}$ and the constant $(3/5)({9\pi}/{2})^{1/3}\simeq 1.45079$. Therefore we have proved that
\begin{multline*}
E_{\rm Ind}\geq -\frac{3}{5}\left(\frac{9\pi}{2}\right)^{1/3}\int\rho(\bx)^{4/3}\,\dx\\+2\int\rho(\bx)D\big(\rho-\rho(\bx),\mu_\bx-\delta_\bx\big)\,\dx,
\end{multline*}
which is our starting point~\eqref{eq:starting}.

\section{Connection between the indirect energy and the classical Jellium problem}\label{app:Wigner}

In this appendix, we will discuss a relation between the Jellium energy at low density and the indirect energy, which is often mentioned in the literature~\cite{ColMar-60,Perdew-91,PerWan-92,LevPer-93,OdaCap-07,RasPitCapPro-09,GorSei-10,RasSeiGor-11}. These two problems look similar but in fact are different. This occasionally causes some confusion and we would like to clarify the difference here.

\subsubsection*{Classical Jellium}

First, consider the \emph{classical Jellium problem}, where there is a constant background density $\rho$ of negative charge in a bounded domain $\Omega$ and there are $N$ classical positive particles with locations $\bx_1,...,\bx_N\in\R^3$. For simplicity we assume global neutrality, in order to make contact with the exchange problem later on. By simple scaling we may assume that $\rho=1$ and we do so henceforth. Thus the total volume of the sample is $|\Omega|=N$. The problem is to choose the domain $\Omega$ and to locate the particles in $\R^3$ so as to minimize the total electrostatic energy
\begin{multline}
\cE_{\rm Jel}(\bx_i,\Omega)=\frac12\sum_{1\leq i\neq j\leq N}\frac1{|\bx_i-\bx_j|}-\sum_{i=1}^N\int_{\Omega}\frac{\dy}{|\bx_i-\by|}\\+\frac12\int_{\Omega}\int_{\Omega}\frac{\dx\dy}{|\bx-\by|},
\label{eq:Jellium_energy_N}
\end{multline}
whose minimum value will be called $E_{\rm Jel}(N)$. 
Nobody knows exactly what this minimum energy in the limit of large $N$ is, but it is believed that the lowest energy configuration is obtained by placing the particles on a BCC lattice. The calculation of this BCC energy is a computational problem that has been done very reliably several times with the result that, in our units with $\rho=1$,
$$\lim_{N\to\ii}\frac{\cE_{\rm Jel}({\rm BCC})}N=-1.4442\cdots.$$
The FCC and cubic lattices give numbers close to 1.4442, see~\cite{ColMar-60} and \cite[p.~43]{GiuVig-05}. 

One thing that is definitely known is that the BCC value must be very close to the true minimum because a rigorous lower bound to the Jellium energy, valid for any domain $\Omega$ and any positions of the particles, was given in \cite{LieNar-75} and reads
$$E_{\rm Jel}(N)\geq -(1.4508\cdots)N,$$
for any finite $N$. The constant is the same as in our Lieb-Oxford bound \eqref{eq:LO_gradient}.

Let us now give an expression of the Jellium energy in the large $N$ limit, for any sufficiently symmetric lattice of density $\rho=1$. Let $\cL$ be any Bravais lattice, with unit cell $Q$ centered at 0 and volume normalized to one. For $Q$ we choose the Wigner-Seitz (or Voronoi) cell, which has the maximum possible symmetry~\cite{AshcroftMermin}. We assume that it has no dipole and no quadrupole:
\begin{equation}
\int_Q \bx\,\dx=0,\qquad \int_Q x_ix_j\,\dx=\frac{\delta_{ij}}3 \int_Q x^2\,\dx. 
 \label{eq:nodipole_quadrupole}
\end{equation}
This is the case for the BCC, FCC and cubic lattices.
We then place the particles on the intersection of the lattice $\cL$ with a large cube $C$ containing $N$ particles and take $\Omega$ to be the union of the cells centered at the particles. After simple manipulations we find that the corresponding Jellium energy in~\eqref{eq:Jellium_energy_N} may be rewritten in the form
\begin{multline}
\cE_{\rm Jel}(\cL\cap C,\Omega)=\frac12\sum_{\bx\neq\by\in \cL\cap C}W(\bx-\by)-\frac{N}2\int_Q\frac{\dy}{|\by|}\\
-\frac12\sum_{\bx,\by\in \cL\cap C}W\ast\1_Q(\bx-\by)
\label{eq:Jellium_energy_N_lattice}
\end{multline}
where
$$W(\bx)=\frac1{|\bx|}-\int_Q\frac{\dy}{|\bx-\by|}$$
is the screened Coulomb potential and where $\ast$ denotes the convolution of two functions:
$$W\ast\1_Q(\bx)=\int_Q W(\bx-\by)\dy.$$
Because there is no dipole and no quadrupole by assumption, the function $W$ decays at infinity at least as fast as $|\bx|^{-4}$ and therefore the limit is readily seen to be
\begin{align}
e_{\rm Jel}(\cL)&:=\lim_{N\to\ii}\frac{\cE_{\rm Jel}(\cL\cap C,\Omega)}N\nn\\
&=\frac12\sum_{\bx\in\cL\setminus\{0\}}W(\bx)-\frac{1}2\int_Q\frac{\dy}{|\by|}
-\frac12\sum_{\bx\in\cL}W\ast\1_Q(\bx)\nn\\
&=\frac12\sum_{\bx\in\cL\setminus\{0\}}W(\bx)-\frac{1}2\int_Q\frac{\dy}{|\by|}
-\frac12\int_{\R^3}W(\by)\,\dy.\label{eq:Jellium_energy}
\end{align}
The value of the limit does not depend on the fact that we have used a large cube $C$ to perform the thermodynamic limit. We would get the same answer by using a large ball or any scaled sufficiently nice set. For the last term in~\eqref{eq:Jellium_energy} we have used the fact that $\sum_{\bx\in\cL}\1_Q(\bx-\by)=1$ for all $\by$ since the unit cells tile the whole space without any intersection. The sum as well as the integral in~\eqref{eq:Jellium_energy} are absolutely convergent, due to the fast decay of $W$. When $\cL$ is the BCC lattice, $e_{\rm Jel}(\cL)\simeq-1.4442$.

\subsubsection*{Indirect Coulomb energy}

Now we turn to the second problem, the indirect Coulomb energy of a probability distribution. Our work here concerns a lower bound to this quantity. While $E_{\rm Ind}$ can be $+\ii$ when the particles are on top of each other, it would be useful to have an upper bound on the minimal possible indirect energy per particle. One way to create an example is to take the (symmetrized) delta function distribution of $N$ particles arranged on the same lattice $\cL$ as before, then shift this whole lattice by an amount $\by$ in $\R^3$ and integrate $\by$ over the unit cell $Q$ of the lattice. For a cubic lattice, we would integrate $\by$ over a cube of side length 1, but for the BCC lattice, we integrate over the Wigner-Seitz cell $Q$, which is a truncated octahedron~\cite{AshcroftMermin}. In this way we construct a probability density $\mathbb{P}(\bx_1,...,\bx_N)$ that is a delta function over the BCC lattice times a uniform measure over the translations of the unit cell of this lattice. This gives a one-body density that is exactly equal to 1 over the region $\Omega$ consisting of the union of $N$ copies of the cells, as discussed before for the Jellium problem. For this $\mathbb{P}$ the expected value of the interparticle Coulomb energy is exactly the Coulomb energy of the lattice $\cL$, and this is the first term in \eqref{eq:Jellium_energy_N}. The self-energy of this one-body density is the third term in \eqref{eq:Jellium_energy_N}. Thus the indirect energy of this state is
\begin{equation}
\cE_{\rm Ind}(\cL\cap C,\Omega)=\frac12\sum_{\bx\neq\by\in \cL\cap C}\frac{1}{|\bx-\by|}-\frac12\int_{\Omega}\int_{\Omega}\frac{\dx\,\dy}{|\bx-\by|}.
\label{eq:indirect_background}
\end{equation}

In the literature, the second term in \eqref{eq:Jellium_energy_N} above is sometimes claimed to be the same quantity as the third term in \eqref{eq:Jellium_energy_N}, up to a factor of 2. In this manner, one is led to think that one has constructed a probability density with an indirect energy that is exactly the same as the Jellium energy, namely $-1.4442$ in the limit $N\to\ii$, if $\cL$ chosen to be the BCC lattice. Unfortunately, \emph{this expectation is not fulfilled}, for the reason we explain now. 

The main difference between the two problems is that the background can be chosen freely in the Jellium problem, whereas it is constrained in the indirect problem by the demand that it must be the one-particle density of the state of the particles. This opens the possibility that surface effects will come into play in the indirect problem.

First, think of the fixed background $\Omega$ and the lattice $\cL$ superposed on it. Then move this lattice a small distance $\by$ while keeping the background fixed. Finally, integrate $\by$ over the unit cell $Q$ of the lattice. The average over $\by$ of the Jellium energy of this shifted lattice is exactly the indirect enery~\eqref{eq:indirect_background}:
\begin{equation}
\cE_{\rm Ind}(\cL\cap C,\Omega)=\int_Q \cE_{\rm Jel}(\cL\cap C+\by,\Omega)\,\dy.
\label{eq:link_indirect}
\end{equation}
This is so because the first and third terms in the Jellium energy do not change, by translation invariance. What does change is the interaction between the particles and the background, and its average coincides with the direct term:
$$\int_Q \dy\left(\sum_{\br\in\cL\cap C}\int_{\Omega}\frac{\dx}{|\br+\by-\bx|}\right)=\int_\Omega\int_\Omega\frac{\dx\dy}{|\bx-\by|}.$$

We see from equation~\eqref{eq:link_indirect} that the indirect energy of our state is an average, over shifted configurations, of the Jellium energy. For $\by=0$, the Jellium energy is minimal. But it is not true that the energy stays the same when $\by\neq0$, because some of the lattice points move closer to the boundary of the domain, thereby raising the Coulomb energy. In the limit $N\to\ii$, this results in a shift which we shall compute explicitly. While it is tempting to suppose that this is a significant effect only for the particles near the boundary, and thus is a negligible effect in the thermodynamic limit, the reality is that this is a volume effect which does not disappear in the limit $N\to\ii$. To appreciate this from a physical perspective, we can think of moving the background instead of moving the lattice. The difference between the moved background and the centered one is a monopole layer on the surface which produces an electric potential, felt by all the particles. By calculating for $\by$ and $-\by$ together, the monopole layer turns into a dipole layer with charge $y^2$ times the area. The effect is a shift of the potential felt by each particle by a constant proportional to $y^2$.

In order to compute the shift, we do not use formula~\eqref{eq:link_indirect}, but rather look at the difference between the indirect and Jellium energies. Using the same notation as in~\eqref{eq:Jellium_energy_N_lattice}, we find for the indirect energy in~\eqref{eq:indirect_background}
\begin{multline}
\cE_{\rm Ind}(\cL\cap C,\Omega)
=\cE_{\rm Jel}(\cL\cap C,\Omega)+\sum_{\bx,\by\in\cL\cap C}W\ast\1_Q(\bx-\by).
\end{multline}
The last term is the one giving rise to the shift and, in this formulation, it clearly is a bulk contribution, since $\bx$ and $\by$ are varied over the whole large cube. Dividing by $N$ and taking the limit we obtain for the lattice $\cL$
\begin{equation}
\lim_{N\to\ii}\frac{\cE_{\rm Ind}(\cL\cap C,\Omega)}{N}=
e_{\rm Jel}(\cL)+\int_{\R^3}W(\by)\,\dy.
\end{equation}
Again the limit does not depend on $C$ being a cube. The shift $\int_{\R^3}W(\by)\,\dy$ can easily be expressed by expanding the Fourier transform of the function $W$ at $\bk=0$ and using the fact that there are no dipoles and quadrupoles, by assumption~\eqref{eq:nodipole_quadrupole}. One gets
\begin{align}
\int_{\R^3}W(\by)\,\dy&=4\pi \lim_{\bk\to0}|\bk|^{-2} \left(1-\int_Qe^{-i\bk\cdot\bx}\,\dx\right)\label{eq:limit_shift}\\
&=\frac{2\pi}{3}\int_Q x^2\,\dx.
\end{align}
Here it is very important that the Fourier transform of the Coulomb potential behaves as $|\bk|^{-2}$ at $\bk=0$. In 2D, the Fourier transform of $1/|\bx|$ is proportional to $1/|\bk|$ and the corresponding limit  in~\eqref{eq:limit_shift} is always $0$. The shift is, therefore, a purely 3D effect.
For the usual Bravais lattices, it is computed to be
\begin{equation}
\frac{2\pi}{3}\int_Q x^2\,\dx=
\begin{cases}
\pi/6=0.5236...&\text{(cubic),}\\
0.4948...&\text{(FCC),}\\
0.4935...&\text{(BCC).}
\end{cases}
\end{equation}
By rearrangement inequalities~\cite{LieLos-01}, the shift is always larger than the same integral on a ball $B$ of unit volume (which, of course, is not a Bravais lattice unit cell):
$$\frac{2\pi}{3}\int_Q x^2\,\dx\geq \frac{3}{10}\left(\frac{4\pi}3\right)^{1/3}\simeq0.4836.$$

The non-vanishing shift for the indirect energy reflects a peculiar feature of the Coulomb potential and what can go wrong when interchanging limits. If one starts with the Yukawa potential $e^{-\nu|\bx|}/|\bx|$, the shift vanishes in the thermodynamic limit for all $\nu>0$:
\begin{multline}
\int_{\R^3}\left(\frac{e^{-\nu|\bx|}}{|\bx|}-\int_{Q}\frac{e^{-\nu|\bx-\by|}}{|\bx-\by|}\,\dy\right)\,\dx\\
=\int_{\R^3}\frac{e^{-\nu|\bx|}}{|\bx|}\,\dx-\int_Q\left(\int_{\R^3}\frac{e^{-\nu|\bx-\by|}}{|\bx-\by|}\,\dx\right)\,\dy=0.
\label{eq:no-shift_Yukawa}
\end{multline}
This can equivalently be seen from~\eqref{eq:limit_shift}, if we replace $|\bk|^{-2}$ by $(\nu^2+|\bk|^2)^{-1}$. If we now take the limit $\nu\to0$ the Yukawa potential converges to the the Coulomb potential and we would incorrectly conclude that there is no shift. The correct limit is obtained by taking $\nu=0$ first and then the thermodynamic limit. By the same argument as in~\eqref{eq:no-shift_Yukawa}, the shift also vanishes when the integral is restricted to a large box and the Coulomb potential is replaced by a periodic function.

The shift discussed here is the same as the one found by D. Borwein, J.M. Borwein, R. Shail and A. Straub in~\cite{BorBorSha-89,BorBorStr-14,BorGlaMPh-13}, although expressed in a different form. The fact that the shift is positive and cannot vanish answers a question raised in~\cite{BorBorStr-14}. 

Another way to understand the appearance of the shift for the indirect energy is as follows. Let us denote by $U$, $I$ and $D$ the Coulomb, the interaction and the direct terms in the Jellium energy~\eqref{eq:Jellium_energy_N}. Then the Jellium energy is $U/2-I+D/2$ whereas the indirect energy is $U/2-D/2$. It is proved in~\cite{BorBorSha-89,BorBorStr-14} that 
$$\frac{U-I}{2N}\to e_{\rm Jel}(\cL)+\frac12\int_{\R^3}W(\bx)\,\dx$$
where $e_{\rm Jel}(\cL)$ is given by the well-known formulas of~\cite[Appendix]{ColMar-60}, that is, $e_{\rm Jel}(\cL)\simeq-1.4442$ when $\cL$ is the BCC lattice. On the other hand, we have seen above that
$$\frac{I-D}{2N}\to \frac12\int_{\R^3}W(\bx)\,\dx,$$
in the limit $N\to\ii$. We observe that the shift disappears for the Jellium energy, but remains for the indirect energy. In the literature, $(U-I)/2N$ is sometimes incorrectly claimed to converge to the Jellium energy $e_{\rm Jel}(\cL)$. On the other hand, if the limit of $(U-I)/2N$ is not carefully computed, the shift is neglected and the right final answer $e_{\rm Jel}(\cL)$ is obtained for the Jellium energy. This wrong mathematical argument is for instance used in~\cite[Appendix A]{ColMar-60}.

Taking the shift into account, the final indirect energy we get with this particular example is
\begin{equation}
E_{\rm Ind}\simeq
\begin{cases}
-0.8950&\text{(cubic),}\\
-0.9494&\text{(FCC),}\\
-0.9507&\text{(BCC).}
\end{cases}
\end{equation}
Although the shift depends on the lattice, the BCC lattice gives the lowest energy.

The conclusion we reach from this calculation is that the indirect Coulomb energy $E_{\rm Ind}$ for a Jellium like density is much less than the Jellium energy, namely $-0.9507$ as against $-1.4442$. This is not a good example for the Lieb-Oxford bound~\eqref{eq:usual_LO} because we already know from a two-particle variational calculation in~\cite[Sec.~3]{LieOxf-80} that the constant is less than $-1.23$, which is way below the Jellium indirect energy. 

The trial state in~\cite[Sec.~3]{LieOxf-80} has a non-vanishing gradient and, therefore, the constant in front of the $\rho^{4/3}$ term in our new inequality~\eqref{eq:LO_gradient} cannot be said to be lower than $-1.23$. However, we can use the previous argument and the fact that the gradient term $\int_{\R^3}|\nabla\rho(\bx)|\,\dx$ is proportional to the area of the boundary of the big domain $\Omega$, which is of lower order than $N$ and disappears in the limit $N\to\ii$. Consequently we conclude that the energy is not less than $-0.9507$. Dealing with the more singular gradient term $\int_{\R^3}|\nabla\rho^{\frac13}(\bx)|^2\,\dx$ would require us to regularize $\rho$, which is beyond the scope of this paper. 

If we summarize the situation, it is rigorously known that the Lieb-Oxford best constant lies somewhere between $-1.23$ and $-1.64$. For constant densities it is between $-0.95$ and $-1.45$ and it is not given by the Jellium indirect energy $-0.95$. Indeed, the constant density case has been numerically studied in~\cite{SeiGorSav-07,RasSeiGor-11}. In~\cite{RasSeiGor-11} a trial state with $N=30$ and an uniform density in a ball was found to have an indirect energy equal to $-1.31$. This says that the constant for the uniform density problem is between $-1.31$ and $-1.45$, but we do not have a conjecture about its value.

After this work was made available on arXiv, our bounds~\eqref{eq:LO_gradient2} were numerically studied in~\cite{FeiKenBur-14,ConFabTerDel-14}. In particular, in~\cite{FeiKenBur-14} it was found that the new bound is not better than the usual Lieb-Oxford bound~\eqref{eq:usual_LO} for spherically symmetric atoms with $Z\leq 88$. Dividing our constant $0.3270$ by a factor two would make it better for all $Z\geq 15$. It is therefore an important challenge to improve all these bounds.

\bigskip

\noindent\textbf{Acknowledgment.} We thank Rupert L. Frank for an interesting discussion, Kieron Burke and John P.~Perdew for useful comments, and the Institut Henri Poincaré in Paris for its hospitality. Grants from the U.S.~NSF PHY-0965859 and PHY-1265118 (E.L.), from the French ANR-10-BLAN-0101  (M.L.) and from the ERC MNIQS-258023 (M.L.) are gratefully acknowledged.


%

\end{document}